\documentclass[11pt]{article}
\usepackage{fullpage}

\usepackage{graphics}
\usepackage[dvips]{epsfig}

\usepackage{amsmath}
\usepackage{amssymb}
\usepackage{amsfonts}
\usepackage{graphicx}

\newtheorem{theorem}{Theorem}[section]
\newtheorem{lemma}{Lemma}[section]
\newtheorem{corollary}{Corollary}[section]

\newtheorem{proposition}{Proposition}[section]

\newcommand{\qed}{\hfill $\Box$ \bigbreak}
\newenvironment{proof}{\noindent {\bf Proof.}}{\qed}

\newcommand{\cA}{{\cal A}}

\newcommand{\remove}[1]{}

\begin{document}

\baselineskip  0.2in 
\parskip     0.1in 
\parindent   0.0in 

\title{{\bf Deterministic Gathering with Crash Faults}}

\author{
Andrzej Pelc\thanks{D\'{e}partement d'informatique, Universit\'{e} du Qu\'{e}bec en Outaouais,
Gatineau, Qu\'{e}bec J8X 3X7,
Canada. E-mail: pelc@uqo.ca.
Supported in part by NSERC discovery grant 8136 -- 2013 
and by the Research Chair in Distributed Computing of
the Universit\'{e} du Qu\'{e}bec en Outaouais.}
}

\date{ }
\maketitle

\begin{abstract}

A team consisting of an unknown number of mobile agents, starting from different nodes of an unknown network, have to meet at the same node and terminate.
This problem is known as {\em gathering}. We study deterministic gathering algorithms under the assumption that agents are subject to {\em crash faults} which can occur at any time.
Two fault scenarios are considered. A {\em motion fault} immobilizes the agent at a node or inside an edge but leaves intact its memory at the time
when the fault occurred. A more severe {\em total fault} immobilizes the agent as well, but also erases its entire memory. Of course, we cannot require faulty agents to gather.
Thus the gathering problem for fault prone agents calls for all fault-free agents to gather at a single node, and terminate.  

When agents move completely asynchronously, gathering with crash faults of any type is impossible. Hence we consider a restricted version of asynchrony, where each agent is assigned by the adversary a fixed speed, possibly different for each agent. Agents have clocks ticking at the same rate. Each agent can wait for a time
of its choice at any node, or decide to traverse an edge but then it moves at constant speed assigned to it. Moreover, agents have different labels. Each agent knows its label and speed but not those of other agents.

We construct a gathering algorithm working for any team of at least two agents in the scenario of motion faults, and a gathering algorithm working in the presence of total faults, provided that at least two agents are fault free all the time. If only one agent is fault free, the task of gathering with total faults is sometimes impossible. Both our algorithms work in time
polynomial in the size of the graph, in the logarithm of the largest label, in the inverse of the smallest speed, and in the ratio between the largest and the smallest speed.

\vspace{2ex}

\end{abstract}

\vfill

\vfill

\thispagestyle{empty}
\setcounter{page}{0}
\pagebreak

\section{Introduction}

\subsection{The background}

A team of at least two mobile entities, called agents, starting from different nodes of a network, have to meet at the same node.
This task is known  as {\em gathering}, and it is called {\em rendezvous} if there are two agents. Both gathering and rendezvous have been extensively studied in the literature under different models.
Mobile entities may represent software agents in computer networks, or mobile robots, if the network is a labyrinth or a cave.
The reason to meet may be to exchange information or ground samples previously collected by the agents,
or to split work in a future task of network exploration or maintenance.
The difficulty of the gathering problem, even in the fault-free context,  comes from the fact that agents
usually do not have any canonical place to meet that is chosen a priori. If agents execute an identical deterministic algorithm, then some symmetry breaking
mechanism must be used to enable meeting, as otherwise meeting would be impossible in highly symmetric networks 
such as, e.g., rings, because agents would make the same moves, always keeping the same distance between them, if they move at the same speed. Since we are interested in deterministic gathering algorithms, we assume that symmetry is broken by assigning a different label to each agent.

 In this paper we study the gathering problem under the assumption that agents are subject to {\em crash faults} which can occur at any time.
Two types of such faults are considered. A {\em motion fault} immobilizes the agent but leaves intact its memory at the time
when the fault occurred. A {\em total fault} immobilizes the agent and erases its entire memory. Of course, we cannot require faulty agents to gather.
Thus the gathering problem for fault-prone agents calls for all fault-free agents to gather at a single node, and terminate.  

When agents move completely asynchronously, gathering with crash faults of any type is impossible. 
This can be seen already in a 3-node clique, each of whose nodes is initially occupied by an agent. Suppose that the adversary lets move agent $a$ at speed 1 but keeps
agents $b$ and $c$ at their initial nodes until agent $a$ terminates. Since each of the agents $b$ and $c$ may have crashed after the last visit of it by agent $a$, at some time agent $a$ must stop at some node and terminate, i.e.,  declare that all fault-free agents have met. This node is different
from the starting node of $b$ or of $c$. At this time the adversary declares that agents $b$ and $c$ are fault-free and lets them go, which makes gathering incorrect, as one of the agents
$b$ or $c$ is not with $a$ at the termination time.

Hence we consider a restricted version of asynchrony,
similar to the one in \cite{DiPe}, where each agent is assigned by the adversary a fixed speed, possibly different for each agent. Agents have clocks ticking at the same rate. Each agent can wait for a time
of its choice at any node, or decide to traverse an edge, but then it moves at constant speed assigned to it.  Each agent knows its label and speed but not those of other agents.

\subsection{The model and the problem}

The network is modeled as a simple undirected connected graph.
As it is often done in the literature on rendezvous and gathering,
we seek gathering algorithms that do not
rely on the knowledge of node labels, and can work in anonymous graphs as well. 
The importance of designing such algorithms
is motivated by the fact that agents may be unable to perceive labels of nodes, even when such labels exist,
because of limitations of perception faculties of the agents, 
or nodes may refuse to reveal their labels, e.g., due to security or privacy reasons.
It should be also noted that, if nodes had distinct labels, agents might explore the graph and meet in the smallest node, hence gathering
would reduce to graph exploration.
On the other hand, we make the assumption, standard in the literature of the domain, that
edges incident to a node $v$ have distinct labels in 
$\{0,\dots,d-1\}$, where $d$ is the degree of $v$. Thus every undirected
edge $\{u,v\}$ has two labels, which are called its {\em port numbers} at $u$
and at $v$. Port numbering is {\em local}, i.e., there is no relation between
port numbers at $u$ and at $v$. Note that in the absence of port numbers, edges incident to a node
would be undistinguishable for agents and thus gathering would be often impossible, 
as the adversary could prevent an agent from taking some edge incident to the current node.
Also, the above mentioned security and privacy reasons that may prevent nodes from revealing their labels, do not apply to port numbers.

Agents start from different nodes of the graph.
They cannot mark visited nodes or traversed edges in any way.
The adversary wakes up some of the agents at possibly different times. A dormant agent is woken up by the first agent that visits
its starting node. 
Agents do not  know the topology of the graph nor any bound on its size. They do not know the total number of agents. They have clocks ticking at the same rate. The clock of each agent starts at its wakeup,
and at this time the agent starts executing the algorithm.

The adversary assigns two numbers to each agent: a positive integer called its {\em label} and a positive real called its {\em speed}.
All labels are different, and if the agent has speed $\alpha$, it traverses each edge in time $1/\alpha$. (We may represent edges as arcs of length 1 between points representing nodes.) Each agent knows its label and its speed, but not those of other agents.
During the execution of an algorithm, an agent can wait at the currently visited node for a time of its choice, or it may choose a port to traverse the corresponding edge.
In the latter case, the agent traverses this edge at the constant speed $\alpha$ assigned to it, thus getting to the other end of the edge after time $1/\alpha$.

When an agent enters a node, it learns its degree and the port of entry. When two or more agents are at the same point in the graph (either at a node or inside an edge)
at the same time, they notice this fact and can exchange all information that is currently in their memory. 
We assume that the memory of the agents is unlimited: from the computational point of view they are modeled as 
Turing machines. 

Agents are subject to {\em crash faults} which can occur at any time.
We consider two fault scenarios. A {\em motion fault} immobilizes the agent at a node or inside an edge but leaves intact its memory at the time
when the fault occurred. A more severe {\em total fault} immobilizes the agent as well, but also erases its entire memory. An agent $A$ meeting an agent $B$
which incurred any crash fault knows that $B$ crashed. In the case of a motion fault of the crashed agent,  $A$ can read the entire memory of $B$ frozen at the time of the crash. By contrast, 
an agent meeting an agent which incurred a total fault cannot learn anything about it, not even its label. An agent that never crashes during the execution of the gathering algorithm is called {\em good}. 

The task of gathering is defined as follows. All good agents have to be at the same node and terminate, i.e., declare that all good agents are at this node. 
(Since good agents located at the same node can mutually read their memory, once one of them declares termination, all others can do it as well, so we may guarantee that the declaration is simultaneous.)
The aim of this paper is the construction of deterministic gathering algorithms under both considered fault scenarios. 

We assume that there is at least one good agent, otherwise gathering of good agents is accomplished vacuously. The time of gathering is counted since the wakeup of the first good agent. Note that counting time since the wakeup of the earliest agent does not make sense, because the adversary can
wake up an agent, crash it immediately, and then wait an arbitrarily long time before waking up the second agent.

\subsection{Our results}

We construct a gathering algorithm working for any team of at least two agents in the scenario of motion faults, and a gathering algorithm working in the presence of total faults, provided that at least two agents are good. If only one agent is good, the task of gathering with total faults is sometimes impossible. 

The main challenge of our scenario is that agents may have different speeds, and at the same time they are subject to crash faults. To the best of our knowledge, these two difficulties together were never considered before in the context of gathering in networks: in \cite{DPV}, asynchronous rendezvous was considered but only for functional agents, while  in \cite{BDD,DPP}, agents with Byzantine faults were considered but only in the fully synchronous scenario, and  \cite{CDLP}  considered
rendezvous of two synchronous agents subject to transient delay faults.

In order to overcome this double difficulty, we adopt a combination of event-driven and time-driven approach.  At early stages of the gathering process,
agents do not know any upper bound on the size of the graph, and hence their decisions cannot rely on time-outs and must be event driven, i.e., they depend on meeting another agent
with a specific memory content.  By contrast, due to our algorithm design, towards the end of gathering, when an agent is already at its final destination, it knows an upper bound on the size of the graph, and it knows everything about other good agents, hence it may compute a time bound after which all good agents must get together.

While our main focus is feasibility, and we do not attempt time optimization, both our algorithms work in time
polynomial in the size of the graph, in the logarithm of the largest label, in the inverse of the smallest speed, and in the ratio between the largest and the smallest speed.


\subsection{Related work}
\label{subsec:relatwork}

The rendezvous problem for two agents was first introduced in \cite{schelling60}. 
An excellent survey of  randomized rendezvous in various scenarios  can be found in
\cite{alpern02b}, cf. also  \cite{alpern95a,alpern02a,anderson90,baston98,israeli}. 
Deterministic rendezvous in networks was surveyed in \cite{Pe}.
In many papers rendezvous was considered in a geometric setting:
an interval of the real line, see, e.g.,  \cite{baston98,baston01,gal99},
or the plane, see, e.g., \cite{anderson98a,anderson98b}).

For the problem of deterministic rendezvous in networks, modeled as undirected graphs, attention concentrated on the study of the feasibility of rendezvous, and on the time required to achieve this task, when feasible. For example, deterministic rendezvous with agents equipped with tokens used to mark nodes was considered, e.g., in~\cite{KKSS}. Deterministic rendezvous of two agents that cannot mark nodes but have unique labels was discussed in \cite{DFKP,KM,TSZ07}.
All these papers were concerned with the time of rendezvous in arbitrary
graphs, under the synchronous scenario. In \cite{DFKP} the authors showed a rendezvous algorithm polynomial in the size of the graph, in the length of the shorter
label and in the delay between the starting time of the agents. In \cite{KM,TSZ07} rendezvous time was polynomial in the first two of these parameters and independent of the delay.

Memory required by the agents to achieve deterministic rendezvous was studied in \cite{FP2} for trees and in  \cite{CKP} for general graphs.
Memory needed for randomized rendezvous in the ring was discussed, e.g., in~\cite{KKPM08}. 

Apart from the synchronous model used in the above papers, several authors investigated asynchronous rendezvous in the plane \cite{CFPS,fpsw} and in network environments
\cite{CLP,DGKKP,DiPe}.
In the latter scenario the agent chooses the edge which it decides to traverse but the adversary controls the speed of the agent. Under this assumption rendezvous
in a node cannot be guaranteed even in very simple graphs, and hence the rendezvous requirement is relaxed to permit the agents to meet inside an edge. In particular, the main result of \cite{DPV} is an asynchronous rendezvous algorithm working in an arbitrary graph at cost (measured by the number of edge traversals) polynomial in the size of the graph and in the logarithm of the smaller label. The scenario of possibly different fixed speeds of the agents, that we use in the present paper, was introduced in \cite{DiPe}, but agents were assumed fault free.

Gathering more than two agents was studied  in the plane \cite{AP, CFPS, DiPe, fpsw}, and in networks \cite{DiPe2,DPP,YY}. In \cite{DiPe2} agents were assumed anonymous (and hence gathering was not always possible), and in  \cite{DPP,YY} they were equipped with distinct labels.

Gathering many agents subject to faults was considered in \cite{AP} in the scenario of anonymous agents in the plane, and in \cite{BDD,CDLP,DPP} in graphs.
In \cite{AP} agents were anonymous but the authors assumed their powerful sensorial capabilities: an agent could see all the other agents in the plane.
In  \cite{BDD,CDLP,DPP} the synchronous model was used. In \cite{CDLP} the faults were assumed relatively benign and transient: the adversary could delay the move of an agent for a few rounds. By contrast, the faults in \cite{BDD,DPP} were of the most severe type,  they were Byzantine, and hence, in order to achieve gathering,
some restrictive upper bound on the number of Byzantine agents had to be imposed. To the best of our knowledge, the problem of gathering in networks,  under the challenging scenario
when agents are not synchronous, and are at the same time subject to permanent faults, was never considered before.

\section{Preliminaries}\label{prelim}

Throughout the paper, 
the number of nodes of a graph is called its size and is denoted by $n$.
In this section we {present three procedures} from the literature, that will be used (after suitable modifications) as building blocks in our algorithms. 
The aim of the first two procedures is graph exploration, i.e., visiting all nodes and traversing all edges of the graph by a single agent. 
The first procedure, based on universal exploration sequences (UXS), is a corollary of the  result of Reingold \cite{Re}. Given any positive integer $n$, it allows the agent to traverse all edges of any graph of size at most $n$,
starting from any node of this graph, using at most $P(n)$ edge traversals, where $P$ is some non-decreasing polynomial. (The original procedure of Reingold only visits all nodes, but it can be transformed to traverse all edges by visiting all neighbors of each visited node before going to the next node.) After entering a node of degree $d$ by some port $p$,
the agent can compute the port $q$ by which it has to exit; more precisely $q=(p+x_i)\mod d$, where $x_i$ is the corresponding term of the UXS.

A {\em trajectory} is a sequence of nodes of a graph, in which each node is adjacent to the preceding one. 
Given any starting node $v$,  we denote by $R(n,v)$ the trajectory obtained by Reingold's procedure. The procedure can be applied in any graph starting at any node, giving
some trajectory. We say that  the agent {\em follows} a trajectory if it executes the above procedure used to construct it.
This trajectory will be called {\em integral}, if the corresponding route covers all edges of the graph. By definition, the trajectory $R(n,v)$ is integral, if it is
obtained by Reingold's procedure applied in any graph of size at most $n$, starting at any node~$v$. 

{The second procedure, described in \cite{DPV}, allows an agent to traverse all edges {and visit all nodes} of any graph of size at most $k$ provided that there is a {unique token located on an {\em extended edge} $u-v$ (for some adjacent nodes $u$ and $v$) and authorized to move arbitrarily on it.
An extended edge is defined as the edge $u-v$ augmented by nodes $u$ and $v$}. (It is well known that a terminating exploration, even of all anonymous rings of unknown size by a single agent without a token, is impossible.) 
In our applications the roles of the token and of the exploring agent will be played by agents. This procedure was called $ESST$ in \cite{DPV}, for {\em exploration with a semi-stationary token}, {as the token always remains on the same extended edge (even if it can move arbitrarily inside it)}. For completeness,  we quote verbatim the description of the procedure from \cite{DPV}.}
The following notion is used in this description. Let $m$ be a positive integer.  An application of $R(2m,u)$ to a graph $G$ at some node $u$ is called {\em clean}, if all nodes in this application are of degree at most $m-1$.

\vspace*{0.2cm}
\noindent
{Procedure {\em  ESST}}

{The procedure proceeds in {phases $i=3,6,9,12\dots$.} In any phase $i$, the agent first applies 
$R(2i,v)$
at the node $v$ in which it started this phase.
Let $(u_1,\dots , u_{r+1})$ be the trajectory $R(2i,v)$ ($v=u_1$ and $r=P(2i)$). Call this trajectory the trunc of this phase.
If it is not clean, or if no token is seen, the agent aborts phase $i$ and starts phase {$i+3$}. Otherwise, the agent backtracks to $u_1$, and applies 
$R(i,u_j)$
at each node $u_j$ of the trunc, interrupting a given execution of $R(i,u_j)$ when it sees a token, every time recording the {\em code} of the path from $u_j$ to this token. 
This code is defined as the sequence of ports
encountered while walking along the path.
(If, for some $j$, the token is at $u_j$, then this code is an empty sequence.) After seeing a token, the agent backtracks to $u_j$, goes
to $u_{j+1}$ and starts executing $R(i,u_{j+1})$. For each node $u_j$ of the trunc, if at the end of $R(i,u_j)$ either no token is seen during the execution of $R(i,u_j)$, or the agent has recorded at least {$\frac{i}{3}$} different codes in phase $i$, then
the agent aborts phase $i$ and starts phase {$i+3$} (in the special case where the agent decides to abort phase $i$ while traversing an edge, phase {$i+3$} starts at the end of this edge traversal). Otherwise, upon completion of phase $i$, it stops as soon as it is at a node.}

The following theorem was proved in \cite{DPV}:

\begin{theorem}\label{esst}
There exists a non-decreasing polynomial $\rho$ such that
procedure $ESST$ terminates in every graph $G$ of size $n$ after $\rho(n)$ steps. Upon its termination, all {edges of $G$ are traversed 
by the agent}.
\end{theorem}

We will use the following corollary of the above theorem.

\begin{corollary}\label{cor}
Consider the procedure $ESST$ performed in some $n$-node graph $G$. There exists a polynomial $\pi(n)$ and integers $\pi_i$, for all $i\leq \pi(n)$, such that the number of phases of $ESST$ is
at most $\pi(n)$, the number of steps in phase $i$ is at most $\pi_i$, and all numbers $\pi_i$ are polynomial in $n$.
\end{corollary}

The third procedure, that will be called {\em Meeting}, is the main rendezvous algorithm from  \cite{DPV}. It works for an arbitrary graph of unknown size, 
under the general asynchronous model, where, in consecutive steps,  the agent chooses the port number that it decides to use, and the adversary produces any continuous movement of the agent on the respective edge, subject only to the constraint that it ends at the other endpoint of the chosen edge.

Procedure {\em Meeting} guarantees the meeting of two agents with distinct labels at cost polynomial in the size of the graph and in the logarithm of the smaller label. 
The following theorem was proved in \cite{DPV}:

\begin{theorem}\label{meeting}
There exists a polynomial $\Pi(x,y)$, non decreasing in each variable, such that
if two agents with different labels $L_1$ and $L_2$ execute Procedure Meeting in a graph of size $n$, then their meeting is guaranteed by the time one of them performs $\Pi(n,\min(\log L_1, \log L_2))$ edge traversals. 
\end{theorem}

The above theorem guarantees that two agents with different labels $L_1$ and $L_2$ executing Procedure Meeting under our present scenario will meet after a total of $\Pi(n,\min(\log L_1, \log L_2))$ edge traversals, even if one of them crashes during the execution of this procedure. This is due to the fact that a walk of an agent under our scenario (at any constant speed and with a possible crash) is an admissible behavior of the agent under the scenario from  \cite{DPV}, where the adversary can hold an agent idle for an arbitrary finite time.

We will use the following notation. For a positive integer $x$, by $|x|$ we denote the length of its binary representation, called the length of $x$. Hence $|x|=\lfloor \log x \rfloor +1$. All logarithms are with base 2.
For two agents, we say that the agent with larger (smaller) label is larger (resp. smaller). 
The length (i.e., the logarithm) of the largest label is denoted by $\Lambda$, the smallest speed of any agent is denoted by $\epsilon$, and the largest speed of any agent is denoted by $\kappa$.
For any trajectory $T=(v_0,\dots,v_r)$, we denote by $\overline{T}$  the reverse trajectory $(v_r,\dots,v_0)$.

\section{Gathering with motion faults}

In this section we construct our gathering algorithm working in the presence of motion faults.
The algorithm works for any team of at least two agents, at least one of which is good. If the team may consist of a single agent that is good, gathering with termination is
impossible, e.g., in rings: a single agent would have to stop and declare termination at some point, while there could be another good agent far away in the ring,
not being able to affect the history of the first agent; hence in the latter scenario, the first  agent would incorrectly declare gathering with termination as well.

  We first describe the high-level idea of the algorithm, skipping the details, and then proceed with the detailed description.. To facilitate
the exposition, we introduce six states in which an agent can be: {\tt cruiser}, {\tt explorer}, {\tt token}, {\tt finder}, {\tt recycled explorer}  and {\tt gatherer}, and describe the goals that an agent achieves in each of these states.

{\bf Algorithm} {\tt Gathering with motion faults}

An agent starts executing the algorithm in state {\tt cruiser}. The goal in this state is to meet another agent, using procedure {\em Meeting}. After meeting another agent in state {\tt cruiser},
agents break symmetry using their labels: the larger agent becomes an {\tt explorer}, the smaller becomes a {\tt token}. The {\tt token} stays put, if they met at a node,
and it continues the current edge  traversal and waits at its endpoint, if they met inside an edge. The {\tt explorer} tries to find an upper bound on the size of the graph.
This is done using a modification of procedure $ESST$, where the agent in state {\tt token} is used as the semi-stationary token. However,  a crucial change
to this procedure has to be made. It must be modified in a way to permit the {\tt token} to realize that the {\tt explorer} crashed, if this is the case. To this end,
after each phase of $ESST$, the {\tt explorer} returns to the {\tt  token}, which permits the latter to use a time-out mechanism (based on the {\tt explorer}'s speed that the {\tt token} learned) that would signal a crash of the {\tt explorer}.
If the {\tt token} realizes that the {\tt explorer} crashed, it transits to state {\tt finder} and tries to meet the crashed {\tt explorer}. It then transits to state {\tt recycled explorer}, performs $ESST$ with this crashed {\tt explorer} as token, in order to
learn the size of the graph.

After learning the size of the graph, an agent transits to state {\tt gatherer}. In this state, an agent
uses procedure {\em Meeting} to meet all other agents, both good and crashed.  It uses the upper bound on the size of the graph to estimate the time after which it can be certain to have met all other agents. This is not straightforward, as the agent does not yet know the speeds of all agents, and it has
to establish the sufficient time of performing procedure $Meeting$ on the basis of currently available information.
When meeting another agent, a {\tt gatherer} records its label, speed and a path towards the starting node of this agent.
After meeting all agents, a {\tt gatherer} goes to the starting node of the smallest agent. Then it uses a time-out mechanism, based on the knowledge of speeds of the other agents. It knows that until some time bound, all other good agents should also be at this node. Then the agent terminates, declaring that gathering of all good agents is accomplished.

The main difficulty of the algorithm design is a combination of event driven and time driven approach. Since at various stages of the algorithm execution the agent has still incomplete knowledge about the graph and about other agents, it cannot rely on time-outs at the early stages. For example, before transiting to state {\tt gatherer} the agent does not know any upper bound on the size of the graph, and hence all transitions must be event driven, e.g., meeting another agent in a given state. By contrast,
towards the end of the execution, when the agent is already at its final destination, it knows an upper bound on the size of the graph, and it knows everything about other good agents, hence it may compute a time bound after which all good agents must get together.

We now give a detailed description of the algorithm, explaining the actions of the agents in each of the six states. In the following description, the executing agent is called $A$, has label $L$ and speed $\theta$.
We assume that at every meeting, agents mutually learn their memory states. This is for convenience only: it will be seen that, in fact, much less information needs to be transferred.

State {\tt cruiser}

Upon wakeup, agent $A$ is in this state. It performs procedure {\em Meeting} until one of the following four events happens: 
\begin{enumerate}
\item
it meets another agent in state {\tt cruiser},
\item
it meets a crashed agent $B$ in state {\tt token},
\item
it meets an agent $C$ in state {\tt finder} or {\tt recycled explorer}
\item
it meets an agent $B$ in state {\tt gatherer}. 
\end{enumerate}
In the first case, let $L_1,\dots ,L_k$ be the labels of agents in state {\tt cruiser} that $A$ met in the first meeting (simultaneous meeting of many agents is possible).
Let $L'=\min(L,L_1,\dots ,L_k)$. If $L=L'$, agent $A$ transits to state {\tt token}, otherwise agent $A$ transits to state {\tt explorer}, and considers the agent with label $L'$ as its token.

In the second case, agent $A$ transits to state {\tt explorer}, and considers agent $B$ as its token. 

In the third case, agent $A$ learns from $C$ the label of the culprit $B$ of $C$ (see the description of state {\tt token}).
Agent $A$ adopts $B$ as its culprit and transits to state {\tt finder}.

In the fourth case, agent $A$ completes the currently executed edge traversal, adopts the upper bound $m$ on the size of the graph from agent $B$,  and transits to state {\tt gatherer}.

State {\tt explorer}

Agent $A$ executes the modified procedure $ESST$, starting at the node of meeting with its token, or, if the meeting occurred inside an edge, completing the current traversal and starting at the corresponding endpoint.  
The procedure is modified in the following way: after each phase (completed or aborted) agent $A$ backtracks to its token following the trajectory $\overline{T}$, where $T$ is the trajectory followed in this phase of $ESST$. Thus the number of phases in the modified Procedure $ESST$ is the same as in $ESST$, and phase $i$ of the modified procedure takes $2\pi_i$ steps (cf. Corollary \ref{cor}). Upon completion of the procedure, the agent traversed all edges of the graph, knows it, and adopts the total number of traversals $m$ as an upper bound on the size of the graph. Agent $A$
gives the integer $m$ to its token and transits to state {\tt gatherer}.

State {\tt token}

Agent $A$ considers all agents that it met when they were in state {\tt cruiser} at the moment of the first meeting (and that transited to state {\tt explorer} upon this meeting) as its {\tt explorers}. Agent $A$  continues its current edge traversal and waits until one of the following two events happens for each of its {\tt explorers}. Either an {\tt explorer} completes the modified procedure $ESST$ giving agent $A$ an upper bound $m$ on the size of the graph, or an {\tt explorer} $B$ is not back after the scheduled time at most  $2\pi_i/\beta$ in some phase $i$, where $\beta$ is the speed of $B$ (learned by $A$). Notice that there may be potentially two different upper bounds given to $A$ by {\tt explorers} that successfully completed
the modified procedure $ESST$ because these {\tt explorers} may start at different endpoints of the edge where they met their {\tt token}, if the meeting occurred inside an edge.
When one of the above events occurs for every {\tt explorer} of the agent $A$, there are two cases. The first case is when at least one {\tt explorer} managed to complete the modified procedure $ESST$ and gave an upper bound to $A$. In this case, agent $A$ adopts the minimum of the (possibly two) obtained upper bounds as its upper bound and transits to state {\tt gatherer}. The second case is when all {\tt explorers} of $A$ failed to  successfully complete the procedure (because they all crashed). Agent $A$ learns this and adopts the (crashed) explorer $B$ with the smallest label among all its explorers as its {\em culprit}.  Agent $A$ transits to state {\tt finder}.

State {\tt finder}

Agent $A$ continues its current edge traversal to the endpoint $v_1$ (if it is at a node $v_1$ when transiting to state {\tt finder}, it skips this action). Then $A$ executes procedures $R(i,v_i)$, for $i=1,2,..., $
one after another,
where $v_{i+1}$ is the node where $R(i,v_i)$ ended, until it finds its culprit. 
Agent $A$ transits to state {\tt recycled explorer}. 

State {\tt recycled explorer}

Agent $A$ executes procedure $ESST$ with its culprit as token. (Here there is no need to modify procedure $ESST$ because the culprit is crashed anyway, so it could not realize that agent $A$ crashed if this happened). Upon completion of $ESST$, agent $A$ adopts the number of steps in this procedure as an upper bound $m$ on the size of the graph. Agent $A$ transits to state {\tt gatherer}.

State {\tt gatherer}

In this state, agent $A$ knows an upper bound $m$ on the size of the graph. Before describing the actions of the agent in this state, we define two procedures.

{\bf Procedure} $Ball(w)$

The procedure consists of traversing all paths of length 2, starting at  node $w$, each time returning to node~$w$. 

The next procedure uses procedure $Ball$.

{\bf Procedure}  $R^*(m,v)$

Follow the trajectory $R(m,v)$, at each visited node $w$ executing procedure $Ball(w)$ before going to the next node of the trajectory.

In state {\tt gatherer}, agent $A$ starts by executing the procedure $R^*(m,v)$, where $v$ is the node in which $A$ transited to this state.
Recall that $\theta$ is the speed of the agent, and let $c$ be the number of edge traversals performed by agent $A$ in the execution of this procedure. Let $\gamma=\frac{\theta}{c}$. We will show that, upon completion of $R^*(m,v)$, agent $A$ has achieved the following two goals: 
it met all agents whose speed is less than $\gamma$, and it learned that all agents are awake. It also follows that at this time, agent $A$ learned the lower bound $\delta=\min(\gamma,\epsilon)$ on the speeds of all agents, where $\epsilon$ is the minimum speed of all agents.

Next, agent $A$ executes procedure $Meeting$ for time $\tau=2\Pi(\nu,\log L)/\delta +3\rho(\nu)/\delta+ \sum_{i=1}^{\nu}P(i)/\delta + \nu^2P(\nu)/\delta$, where $\nu=2\rho(m)$.
Call this period of time the {\em red} period of agent $A$.

We will show that after this time, agent $A$ has met all the agents, and hence it knows their labels,
their speeds,  and the path to the starting node of each of them. Moreover, it knows the largest upper bound $\mu$ on the size of the graph, among those adopted by all the agents.
Agent $A$ goes to the starting node of the agent whose label is the smallest, and starts its final waiting period.

Let $\nu^*=2\rho(\mu)$. $\nu^*$ is an upper bound on the values of $\nu$ calculated by all agents. Let $\delta ^*$ be the minimum of the values of $\delta$ calculated by all agents.
 Let $\phi=21\rho(\nu^*)/\delta^*+7 \sum_{i=1}^{\nu^*}P(i)/\delta^*+8\Pi(\nu^*,\Lambda)/\delta^*+8(\nu^*)^2P(\nu^*)/\delta^*$. We will show that the value $\phi$ is an upper bound on the termination of the red periods of all agents.
 
The final waiting period of agent $A$ has length $\tau'=\phi+2\phi\kappa/\epsilon$.

Notice that this final waiting period has the same length for all agents, as it does not depend on the label and the speed of an individual agent.
Agent $A$ can calculate this waiting time because, as we will show,  when it starts it, it knows all parameters concerning other agents.

We will prove that after this time all good agents must also be at the starting node of the smallest agent, and that agent $A$ learns this fact. After waiting time $\tau'$ (or at the time when some other agent declares gathering, if this happens earlier), agent $A$ declares that all good agents are gathered, and terminates.

The proof of correctness and the performance analysis of Algorithm  {\tt Gathering with motion faults} are split into a series of lemmas.

\begin{lemma}\label{first}
By the time $\Pi(n,\Lambda)/\epsilon$ some agents must meet.
\end{lemma}

\begin{proof}
Let $D$ be the earliest woken good agent (or one of them, if there are several such agents), and let $L_D$ be its label.
By Theorem \ref{meeting}, $D$ must meet some agent by the time when $D$ performs $\Pi(n,\log L_D)$ edge traversals executing procedure $Meeting$ in state {\tt cruiser}, if no other agents have met before. This must happen by the time $\Pi(n,\log L_D)/\epsilon \leq \Pi(n,\Lambda)/\epsilon$.
\end{proof}

Let $S_0$ be the earliest woken good agent. (If there are more such agents, we denote by $S_0$ any of them.) By definition, time is counted since the wakeup of agent $S_0$, i.e., its wakeup time is set to 0. 

\begin{lemma}\label{gatherer}
By the time $6\rho(n)/\epsilon+2 \sum_{i=1}^nP(i)/\epsilon+2\Pi(n,\Lambda)/\epsilon$ some agent transits to state {\tt gatherer}.
\end{lemma}

\begin{proof}
Let $t_1$ be the time when the first meeting occurred. This was necessarily a meeting of agents in state {\tt cruiser}. First suppose that one of these agents, call it $F$, was good. If $F$ transits to state {\tt explorer} upon this meeting, then by the time $t_1+2\rho(n)/\epsilon$ agent $F$ completes the modified procedure $ESST$, and transits to state {\tt gatherer}.  If $F$ transits to state {\tt token} upon this meeting, then by the time $t_1+2\rho(n)/\epsilon$ agent $F$ either learns
an upper bound on the size of the graph (because some of its {\tt explorers} completed the modified procedure $ESST$) and transits to state {\tt gatherer}, or it learns that all its {\tt explorers} crashed. In the latter case, agent $F$ transits to state {\tt finder}, finds its culprit $B$ within an additional time $ \sum_{i=1}^nP(i)/\epsilon$ (because within this time it traverses all edges of the graph), transits to state {\tt recycled explorer}, then within an additional time $\rho(n)/\epsilon$, it completes procedure $ESST$ with $B$ as token, learns an upper bound on the size of the graph, and transits to state {\tt gatherer}. 
Hence, if one of the agents participating in the first meeting was good, then this agent must transit to state {\tt gatherer} by the time $t_1+3\rho(n)/\epsilon+  \sum_{i=1}^nP(i)/\epsilon$.

It remains to consider the case, when none of the agents participating in the first meeting was good.
If the agent $T$ that transited to state {\tt token} at the first meeting did not crash by the time $t_1+3\rho(n)/\epsilon+  \sum_{i=1}^nP(i)/\epsilon$, then the above argument shows that it transits to state {\tt gatherer} by this time. Hence we may assume that this agent crashed by the time $t_1+3\rho(n)/\epsilon+  \sum_{i=1}^nP(i)/\epsilon$.

 Let $t_2$ be the time when agent $S_0$ transits from state {\tt cruiser} to some other state. This happened because of one of the
four events in the definition of state {\tt cruiser}. 

If in time $t_2$ agent $S_0$ met another agent in state {\tt cruiser}, then it transited to state {\tt gatherer}  by the time $t_2+3\rho(n)/\epsilon+  \sum_{i=1}^nP(i)/\epsilon$,
in view of the argument above.

If in time $t_2$ agent $S_0$ met a crashed agent in state {\tt token}, then it transited to state {\tt explorer}, and by the time $t_2 +2\rho(n)/\epsilon$ it completes the modified procedure $ESST$, and transits to state {\tt gatherer}.

If  in time $t_2$ agent $S_0$ met an agent $C$ in state {\tt finder} or {\tt recycled explorer} then it learns from $C$ the label of the culprit $B$ of $C$, 
adopts $B$ as its culprit and transits to state {\tt finder}. It finds its culprit $B$ within an additional time $ \sum_{i=1}^nP(i)/\epsilon$, transits to state {\tt recycled explorer}, then within an additional time $\rho(n)/\epsilon$ it completes procedure $ESST$ with $B$ as token, learns an upper bound on the size of the graph, and transits to state {\tt gatherer} by the time $t_2+\rho(n)/\epsilon+  \sum_{i=1}^nP(i)/\epsilon$.

If  in time $t_2$ agent $S_0$ met an agent in state {\tt gatherer}, then it learns from it an upper bound on the size of the graph, and immediately transits to state {\tt gatherer}.

Hence, in all cases, agent $S_0$ transits to state gatherer by the time $t_2+3\rho(n)/\epsilon+  \sum_{i=1}^nP(i)/\epsilon$. We now estimate the time $t_2$
when agent $S_0$ transits from state {\tt cruiser} to some other state. By our assumption, agent $T$ crashed by the time $t_1+3\rho(n)/\epsilon+  \sum_{i=1}^nP(i)/\epsilon$. In the moment of crashing, it could be either still in the state {\tt token}, or in state {\tt finder} or in state {\tt recycled explorer} (otherwise, it would transit to state {\tt gatherer}
before crashing, and the lemma is proved). Consider the time
$t_1+3\rho(n)/\epsilon+  \sum_{i=1}^nP(i)/\epsilon+\Pi(n,\Lambda)/\epsilon$. If by this time agent $S_0$ is still in state {\tt cruiser}, then by this time it must meet agent $T$
in one of these states because $S_0$ performed at least $\Pi(n,\Lambda)$ edge traversals executing procedure $Meeting$, while $T$ was still. Hence $S_0$ must transit from state {\tt cruiser} to some other state by this time. It follows that $t_2 \leq t_1+3\rho(n)/\epsilon+  \sum_{i=1}^nP(i)/\epsilon+\Pi(n,\Lambda)/\epsilon$.

Finally, we have $t_1 \leq \Pi(n,\Lambda)/\epsilon$, by Lemma \ref{first}. We conclude that some agent transits to state {\tt gatherer} by the time
$t_2+3\rho(n)/\epsilon+  \sum_{i=1}^nP(i)/\epsilon \leq  t_1+3\rho(n)/\epsilon+  \sum_{i=1}^nP(i)/\epsilon+\Pi(n,\Lambda)/\epsilon+3\rho(n)/\epsilon+  \sum_{i=1}^nP(i)/\epsilon \leq  6\rho(n)/\epsilon+2 \sum_{i=1}^nP(i)/\epsilon+2\Pi(n,\Lambda)/\epsilon$.
\end{proof}

Let $m$ be the upper bound on the size of the graph, adopted by the first agent that transited to state {\tt gatherer}.

\begin{lemma}\label{good gatherer}
Let $S$ be a good agent, and let $t$ be its waking time.
By the time $t+9\rho(n)/\epsilon+3 \sum_{i=1}^nP(i)/\epsilon+3\Pi(n,\Lambda)/\epsilon+2m^2P(m)/\epsilon$ agent $S$ transits to state {\tt gatherer}.
\end{lemma}

\begin{proof}
Let $t'$ be the time when the first agent, call it $H$, transits to state {\tt gatherer}. In view of Lemma \ref{gatherer}, we have $t'\leq 6\rho(n)/\epsilon+2 \sum_{i=1}^nP(i)/\epsilon+2\Pi(n,\Lambda)/\epsilon$. Let $m$ be the upper bound on the size of the graph, adopted by agent $H$. 
In the execution of procedure $R^*(m,v)$ agent $H$ performed $c\leq 2m^2P(m)$ edge traversals. Hence by the time $t'+2m^2P(m)/\epsilon$, it completed procedure
$R^*(m,v)$ and started procedure $Meeting$, if it did not crash before. Hence agent $S$ must  transit from state {\tt cruiser} to some other state
by the time $t+ t'+2m^2P(m)/\epsilon+\Pi(n,\Lambda)/\epsilon$. Indeed, if it did not transit  by the time $t+t'+2m^2P(m)/\epsilon+\Pi(n,\Lambda)/\epsilon$, then it 
performed at least $\Pi(n,\Lambda)$ edge traversals executing procedure $Meeting$ while agent $H$ was also executing this procedure, or else $H$ crashed in the meantime. Hence $S$ must have met $H$ and transited to state {\tt gatherer}.

Let $t_2$ be the time when agent $S$ transits from state {\tt cruiser} to some other state. It follows from the above that 
$t_2 \leq t+ t'+2m^2P(m)/\epsilon+\Pi(n,\Lambda)/\epsilon$. By the argument from the proof of Lemma \ref{gatherer}, agent $S$ must transit to state {\tt gatherer} by the time
$t_2+3\rho(n)/\epsilon+  \sum_{i=1}^nP(i)/\epsilon$. Hence agent $S$ must transit to state {\tt gatherer} by the time
$$t_2+3\rho(n)/\epsilon+  \sum_{i=1}^nP(i)/\epsilon\leq$$ 
$$t+t'+2m^2P(m)/\epsilon+\Pi(n,\Lambda)/\epsilon+3\rho(n)/\epsilon+  \sum_{i=1}^nP(i)/\epsilon\leq$$
$$t+6\rho(n)/\epsilon+2 \sum_{i=1}^nP(i)/\epsilon+2\Pi(n,\Lambda)/\epsilon+2m^2P(m)/\epsilon+\Pi(n,\Lambda)/\epsilon+3\rho(n)/\epsilon+  \sum_{i=1}^nP(i)/\epsilon=$$
$$t+9\rho(n)/\epsilon+3 \sum_{i=1}^nP(i)/\epsilon+3\Pi(n,\Lambda)/\epsilon+2m^2P(m)/\epsilon.$$
\end{proof}

Let $\mu$ be the largest of the upper bounds on the size of the graph adopted by any agent.

\begin{lemma}\label{woken}
By the time $9\rho(n)/\epsilon+3 \sum_{i=1}^nP(i)/\epsilon+3\Pi(n,\Lambda)/\epsilon+4\mu^2P(\mu)/\epsilon$ all agents are woken up.
\end{lemma}

\begin{proof}
In view of Lemma \ref{good gatherer}, agent $S_0$ transits to state {\tt gatherer} by the time $9\rho(n)/\epsilon+3 \sum_{i=1}^nP(i)/\epsilon+3\Pi(n,\Lambda)/\epsilon+2\mu^2P(\mu)/\epsilon$. Hence, by the time $9\rho(n)/\epsilon+3 \sum_{i=1}^nP(i)/\epsilon+3\Pi(n,\Lambda)/\epsilon+4\mu^2P(\mu)/\epsilon$, agent $S_0$ completed procedure $R^*(m_1,v)$, where $v$ is the node at which $S_0$ transited to state {\tt gatherer}. Hence it traversed all edges of the graph, an thus
all agents are woken up by this time.
\end{proof}

\begin{lemma}\label{all gatherer}
All good agents transit to state {\tt gatherer} by the time $18\rho(n)/\epsilon+6 \sum_{i=1}^nP(i)/\epsilon+6\Pi(n,\Lambda)/\epsilon+6\mu^2P(\mu)/\epsilon$.
\end{lemma}

\begin{proof}
Let $S$ be a good agent and let $t$ be its waking time. By Lemma \ref{good gatherer}, $S$ transits to state {\tt gatherer} by the time
$t+9\rho(n)/\epsilon+3 \sum_{i=1}^nP(i)/\epsilon+3\Pi(n,\Lambda)/\epsilon+2\mu^2P(\mu)/\epsilon$. By Lemma \ref{woken},
$t \leq 9\rho(n)/\epsilon+3 \sum_{i=1}^nP(i)/\epsilon+3\Pi(n,\Lambda)/\epsilon+4\mu^2P(\mu)/\epsilon$.
Hence $S$ transits to state {\tt gatherer} by the time
$18\rho(n)/\epsilon+6 \sum_{i=1}^nP(i)/\epsilon+6\Pi(n,\Lambda)/\epsilon+6\mu^2P(\mu)/\epsilon$.
\end{proof}

\begin{lemma}\label{red}
Any good agent $S$ meets all other agents by the time it executes its red period. This happens by the time $21\rho(\nu)/\delta+7 \sum_{i=1}^{\nu}P(i)/\delta+8\Pi(\nu,\Lambda)/\delta+8\nu^2P(\nu)/\delta$.
\end{lemma}

\begin{proof}
Let $m(S)$ be the upper bound on the size of the graph, adopted by agent $S$,  let $\theta$ be the speed of this agent, and let $L$ be its label.
After executing procedure $R^*(m(S),v)$, where $v$ is the node where $S$ transited to state {\tt gatherer}, agent $S$ traversed all edges, and hence it knows that all other agents are woken up. During the execution of $R^*(m,v)$ agent $S$ makes $c\leq 2m(S)^2P(m(S))$ edge traversals. Also,  during this execution, agent $S$ traverses all directed paths of length 2 in the graph.
Let $\gamma=\frac{\theta}{c}$, and consider an agent $S'$ whose speed is less than $\gamma$. 
The time $c/\theta$ during which agent $S$ makes the $c$ edge traversals executing procedure $R^*(m(S),v)$ is too short for agent $S'$ to make even one full edge traversal. Hence, regardless of the trajectory followed by agent $S'$, it can be  in at most two incident edges during this time. Since agent $S$ traverses all such couples of edges, it must meet agent $S'$ during the execution of $R^*(m(S),v)$. It follows that after the completion of procedure $R^*(m(S),v)$, agent $S$ either learns the minimum speed $\epsilon$ of all agents (if there exists an agent with speed less than $\gamma$) or it learns that $\gamma$ is a lower bound on the speed of all the agents (if there is no agent with speed less than $\gamma$). Set $\delta$ to $\epsilon$ in the first case and to $\gamma$ in the second case.
Notice that, upon completion of procedure $R^*(m(S),v)$, agent $S$ learns the value of $\delta$.

Consider any good agent $S''$, with label $L''$, whose speed is at least $\gamma$.  Upon completion of procedure $R^*(m(S),v)$ at time $t_1$, agent $S$ executes procedure $Meeting$ and agent $S''$ is already woken up.
We will prove that agent $S$ meets agent $S''$ by the time $t_1+\tau$. We may assume that by the time $t_2=t_1+\Pi(m(S),\log L)/\delta$, agent $S''$ transited from state {\tt cruiser} to some other state,
for otherwise it would execute procedure $Meeting$ during the time when $S$ performed at least $\Pi(m(S),\log L)$ edge traversals executing this procedure, and hence the agents would meet, in view of  Theorem
\ref{meeting}.
This transition of agent $S''$ happened because of one of the
four events in the definition of state {\tt cruiser}. 

If  agent $S''$ met another agent in state {\tt cruiser}, then it either transited to state {\tt explorer} or to state {\tt token}.
If $S''$ transits to state {\tt explorer} upon this meeting, then by the time $t_2+2\rho(n)/\gamma$ agent $S''$ completes the modified procedure $ESST$, and transits to state {\tt gatherer}.  If $S''$ transits to state {\tt token} upon this meeting, then by the time $t_2+2\rho(n)/\gamma$ agent $S''$ either learns
an upper bound on the size of the graph (because some of its {\tt explorers} completed the modified procedure $ESST$) and transits to state {\tt gatherer}, or it learns that all its {\tt explorers} crashed. In the latter case, agent $S''$ transits to state {\tt finder}, finds its culprit $B$ within an additional time $ \sum_{i=1}^nP(i)/\gamma$ (because within this time it traverses all edges of the graph), transits to state {\tt recycled explorer}, then within an additional time $\rho(n)/\gamma$, it completes procedure $ESST$ with $B$ as token, learns an upper bound on the size of the graph, and transits to state {\tt gatherer}. 
Hence, if agent $S''$ met an agent in state {\tt cruiser}, then $S''$ must transit to state {\tt gatherer} by the time $t_2+3\rho(n)/\gamma+  \sum_{i=1}^nP(i)/\gamma$.

If  agent $S''$ met a crashed agent in state {\tt token}, then it transited to state {\tt explorer}, and by the time $t_2 +2\rho(n)/\gamma$ it completes the modified procedure $ESST$, and transits to state {\tt gatherer}.

If  agent $S''$ met an agent $C$ in state {\tt finder} or {\tt recycled explorer} then it learns from $C$ the label of the culprit $B$ of $C$, 
adopts $B$ as its culprit and transits to state {\tt finder}. It finds its culprit $B$ within an additional time $ \sum_{i=1}^nP(i)/\gamma$, transits to state {\tt recycled explorer}, then within an additional time $\rho(n)/\gamma$ it completes procedure $ESST$ with $B$ as token, learns an upper bound on the size of the graph, and transits to state {\tt gatherer} by the time $t_2+\rho(n)/\gamma+  \sum_{i=1}^nP(i)/\gamma$.

If  agent $S''$ met an agent in state {\tt gatherer}, then it learns from it an upper bound on the size of the graph, and immediately transits to state {\tt gatherer}.

Hence, in all cases, agent $S''$ transits to state {\tt gatherer} by the time $t_3=t_2+3\rho(n)/\gamma+  \sum_{i=1}^nP(i)/\gamma \leq t_1+ \Pi(m(S),\log L)/\delta+ 3\rho(m(S))/\delta+  \sum_{i=1}^{m(S)}P(i)/\delta$. 

Let $m(S'')$ be the upper bound on the size of the graph adopted by agent $S''$ and let $v''$ be the node at which $S''$ transited to state {\tt gatherer}. Upon the transition, agent $S''$ executes procedure $R^*(m(S''),v'')$ within time at most  $m(S'')^2P(m(S''))/\gamma$ and then starts its own red period executing procedure $Meeting$.  By the time $t_4=t_3+ m(S'')^2P(m(S''))/\gamma+\Pi(m(S),\log L)/\delta$
agent $S$ must meet agent $S''$ because either $S$ performed at least $\Pi(m(S),\log L)$ edge traversals executing procedure $Meeting$, while $S''$ also executed procedure $Meeting$, or $S''$ performed at least $\Pi(m(S),\log L'')$ edge traversals executing procedure $Meeting$, while $S$ also executed procedure $Meeting$.

Hence by the time $t_4$, agent $S$ must meet every good agent whose speed is at least $\gamma$.
 Note that if agent $S''$ is not good and (possibly) crashed at some point, its behavior (standing still since the crash) is also a possible behavior of an asynchronous agent during procedure $Meeting$, and hence $S$ must meet every agent whose speed is at least $\gamma$ by time $t_4$, if $S$ was still executing procedure $Meeting$.  
 
 Since $S$ already met all 
 agents whose speed is less than $\gamma$ by the time $t_1$, we conclude that by the time $t_4$
 agent $S$ must meet all agents, if it executed procedure $Meeting$ during its red period for a sufficiently long time. 
 
 Consider agent $S$ upon completion of procedure $R^*(m(S),v)$. Agent $S$ knows $\delta$ and knows some upper bound on the largest upper bound $\mu$ on the size of the graph adopted by any agent.
 Indeed,  any agent $T$ adopts as its upper bound the number of edge traversals in its execution of procedure $ESST$. This number is at most $2\rho(n)$. Hence it is at most $\nu=2\rho(m(S))$, which $S$ knows.  Thus $\mu \leq \nu$.
 
 It follows that $t_4=t_3+ m(S'')^2P(m(S''))/\gamma+\Pi(m(S),\log L)/\delta \leq t_1+ \Pi(m(S),\log L)/\delta+ 3\rho(m(S))/\delta+  \sum_{i=1}^{m(S)}P(i)/\delta + m(S'')^2P(m(S''))/\gamma+\Pi(m(S),\log L)/\delta
 \leq t_1+2\Pi(\nu,\log L)/\delta +3\rho(\nu)/\delta+ \sum_{i=1}^{\nu}P(i)/\delta + \nu^2P(\nu)/\delta=t_1+\tau$.
 
 Hence by the end of its red period in time $t_1+\tau$, agent $S$ meets all the agents. Since agent $S$ transits to state {\tt gatherer} by the time $18\rho(n)/\epsilon+6 \sum_{i=1}^nP(i)/\epsilon+6\Pi(n,\Lambda)/\epsilon+6\mu^2P(\mu)/\epsilon$, in view of Lemma \ref{all gatherer}, we have
 $t_1 \leq  18\rho(n)/\epsilon+6 \sum_{i=1}^nP(i)/\epsilon+6\Pi(n,\Lambda)/\epsilon+6\mu^2P(\mu)/\epsilon+ \nu^2P(\nu)/\delta \leq
 18\rho(\nu)/\delta+6 \sum_{i=1}^{\nu}P(i)/\delta+6\Pi(\nu,\Lambda)/\delta+7\nu^2P(\nu)/\delta$.
 
 Since $\tau=2\Pi(\nu,\log L)/\delta +3\rho(\nu)/\delta+ \sum_{i=1}^{\nu}P(i)/\delta + \nu^2P(\nu)/\delta$, the end of the red period of agent $S$ occurs by the time
$21\rho(\nu)/\delta+7 \sum_{i=1}^{\nu}P(i)/\delta+8\Pi(\nu,\Lambda)/\delta+8\nu^2P(\nu)/\delta$.
\end{proof}

Let $\nu^*=2\rho(\mu)$. $\nu^*$ is an upper bound on the values of $\nu$ calculated by all agents. Let $\delta ^*$ be the minimum of the values of $\delta$ calculated by all agents.
 Let $\phi=21\rho(\nu^*)/\delta^*+7 \sum_{i=1}^{\nu^*}P(i)/\delta^*+8\Pi(\nu^*,\Lambda)/\delta^*+8(\nu^*)^2P(\nu^*)/\delta^*$. The value $\phi$ is an upper bound on the termination of the red periods of all agents.
 
 \begin{lemma}\label{destin all}
 By the time $\phi+2\phi\kappa/\epsilon$ any good agent $S$ is at the starting node of the smallest agent $S_{min}$ and starts its final waiting period. 
 \end{lemma}
 
 \begin{proof}
 Let $T$ be the trajectory of agent $S$ until the end of its red period, and let $T'$ be the trajectory of agent $S_{min}$
 until the first meeting with agent $S$. The length of each of the trajectories $T$ and $T'$ is at most $\phi\kappa$. After meeting all agents, agent $S$ knows both these trajectories. Hence it can backtrack to the meeting point and then backtrack to the starting node of agent $S_{min}$. This takes at most $2\phi\kappa$ edge traversals, and hence it takes time at most $2\phi\kappa/\epsilon$. Hence by the time 
 $\phi+2\phi\kappa/\epsilon$
 agent $S$ is at the starting node of the smallest agent and starts its final waiting period.
  \end{proof}
  
   \begin{lemma}\label{all}
By the time $2\phi+4\phi\kappa/\epsilon$  all good agents are at the starting node of the smallest agent, 
all of them know it, and correctly declare that gathering is completed.
\end{lemma}

\begin{proof}
We first show that at the beginning of its final waiting period, every good agent can compute the length $\tau'=\phi+2\phi\kappa/\epsilon$ of this period.
Indeed this value depends on five parameters: $\Lambda$, $\kappa$, $\epsilon$, $\nu^*$ and $\delta^*$.
At the beginning of its final waiting period, every good agent has met all agents and knows all their parameters. Hence it can compute the largest length $\Lambda$ of all labels, 
the largest speed $\kappa$,  and the smallest speed $\epsilon$. It learns all values $\delta$ computed by all agents, and finds their minimum $\delta ^*$. 
It learns all values $\nu$ computed by all agents, and finds their maximum $\nu^*$. Hence it can compute $\phi$ and finally $\tau'$.

In view of Lemma \ref{destin all}, by the time $\phi+2\phi\kappa/\epsilon$, every good agent is at the starting node of the smallest agent and starts its final waiting period.
Hence at the end of its final waiting period, which has the same length $\tau'=\phi+2\phi\kappa/\epsilon$, every good agent knows that all good agents started their final waiting period.
Hence every good agent correctly declares that gathering is completed by the time $2\phi+4\phi\kappa/\epsilon$.
\end{proof}

We are now ready to prove the main result of this section.

\begin{theorem}
Algorithm {\tt Gathering with motion faults} is a correct algorithm for the task of gathering with motion faults. It works in time polynomial in the size $n$ of the graph, in the length $\Lambda$ of the largest label,
 in the inverse of the smallest speed $\epsilon$, and in the ratio $r$ between the largest and the smallest speed.
\end{theorem}

\begin{proof}
The correctness of the algorithm follows immediately from Lemma \ref{all}. 
It remains to verify that the time $2\phi+4\phi\kappa/\epsilon$  is polynomial in $n$, $\Lambda$, $1/\epsilon$ and $r$. Since $r=\kappa/\epsilon$, it is enough to verify that 
$\phi=21\rho(\nu^*)/\delta^*+7 \sum_{i=1}^{\nu^*}P(i)/\delta^*+8\Pi(\nu^*,\Lambda)/\delta^*+8(\nu^*)^2P(\nu^*)/\delta^*$ is polynomial in  $n$, $\Lambda$ and $1/\epsilon$.

The upper bound $m$ on the size of the graph, calculated by each agent, is at most $2\rho(n)$, in view of Lemma \ref{esst}.
Since $\rho$ is a polynomial, it follows that $m$ is polynomial in $n$, and so is $\mu$. Hence $\nu^*=2\rho(\mu)$ is also polynomial in $n$.

Since the value $\gamma$ computed by each agent is $\frac{\theta}{c}$, where $\theta$ is the speed of the agent and $c\leq \mu^2P(\mu)$, it follows that $1/\gamma$ is polynomial in
$n$ and in $1/\epsilon$. Hence the value $1/\delta$ computed by each agent is  polynomial in
$n$ and in $1/\epsilon$, and thus so is $1/\delta ^*$.

It follows that $21\rho(\nu^*)/\delta^*$ is polynomial in
$n$ and in $1/\epsilon$. Since $P$ is a polynomial, it follows that $7 \sum_{i=1}^{\nu^*}P(i)/\delta^*$ and $8(\nu^*)^2P(\nu^*)/\delta^*$ are polynomial in
$n$ and in $1/\epsilon$. Since $\Pi$ is a polynomial, it follows that $8\Pi(\nu^*,\Lambda)/\delta^*$ is polynomial in $n$, $\Lambda$,
and $1/\epsilon$. Hence $\phi$ is polynomial in $n$ $\Lambda$,
and $1/\epsilon$.
\end{proof}

\section{Gathering with total faults}

In this section we construct a gathering algorithm in the scenario of total faults, where an agent meeting a crashed agent cannot learn anything about it: all crashed agents look identical.
We first show that in the presence of these more severe faults, gathering may be impossible if only one agent is good. (Gathering in this case means that this agent should stop and declare that
all good agents are with it, i.e., that it is the only good agent.) This is in contrast with the motion faults scenario, in which our algorithm worked correctly for any team of at least two agents, even if only one agent was good.

\begin{proposition}
Gathering with total faults cannot be guaranteed if there may be only one good agent.
\end{proposition}

\begin{proof}
Consider the class of rings of size $4k$, for all positive integers $k$, with port numbers 0 and 1 at every edge. This port numbering orients the rings.
Suppose that there exists an algorithm $\cA$ that accomplishes the gathering in this class of rings.  
In the ring of size 4, the adversary places agents in two antipodal nodes and crashes them immediately. Moreover, it places a good agent in one of the remaining nodes, assigning it the speed 1.
If this single good agent executing $\cA$ has label 1, it must stop after $t_1$ steps and declare that it is the unique good agent.
If it has label $2$, it must stop after $t_2$ steps and declare that it is the unique good agent.

Consider the ring of size $4(t_1+t_2)$, in which the adversary places agents at every second node, and crashes them immediately. Moreover, it places two good agents at two antipodal nodes that are not yet occupied, assigns speed 1 to each of them, assigns label 1 to one of them, label 2 to the other, and starts them simultaneously. Consider the agent with label 1 executing the algorithm for $t_1$ steps, and the agent with
label 2 executing the algorithm for $t_2$ steps. By this time, these agents cannot meet because initially they are at distance $2(t_1+t_2)$. Hence the history of each of them is the same as in the ring of size 4.
Consequently, since the algorithm $\cA$ is deterministic, each of them must make the same decision as in the ring of size 4 where it was the unique good agent, i.e., stop and declare gathering
(declare that it is the unique good agent). Since in the considered scenario there are at this time two good agents at different nodes, algorithm $\cA$ is incorrect.
\end{proof}

Hence, in the rest of this section, we assume that there are at least two good agents. This means that the task of gathering with total faults is on the one hand more difficult because a good agent cannot learn anything from a crashed agent, but on the other hand it is easier because a good agent can rely on the existence of another good agent, and plan its moves accordingly.

As before, we first describe the high-level idea of the algorithm, pointing out the differences with respect to the Algorithm {\tt Gathering with motion faults},  and then proceed with the detailed description. We will now have only four states in which an agent can be: {\tt cruiser}, {\tt explorer}, {\tt token},  and {\tt gatherer}. 

{\bf Algorithm} {\tt Gathering with total faults}

An agent starts executing the algorithm in state {\tt cruiser}. The goal in this state is to meet another agent, using procedure {\em Meeting}. 
However, as opposed to the previous algorithm, an agent in state {\tt cruiser} executes this procedure until meeting another agent in state {\tt cruiser} or {\tt gatherer}, ignoring all meetings with agents in
states {\tt token} or {\tt explorer} and ignoring meetings with crashed agents. As opposed to the previous scenario, we can afford to ignore these meetings because there are at least two good agents.

After meeting an agent in state {\tt cruiser},
agents break symmetry using their labels: the larger agent becomes an {\tt explorer}, the smaller becomes a {\tt token}. The {\tt explorer} tries to find an upper bound on the size of the graph.
This is done using the modified procedure $ESST$, where the agent in state {\tt token} is used as the semi-stationary token. Now there are two modifications of $ESST$. One is as before, returning
to the token after each phase, and it is made for the same reason. The second modification is that when the {\tt explorer} meets a crashed agent, it returns to the {\tt token} to verify if this {\tt token} did not crash.
This has to be done because all crashed agents look the same, so the {\tt explorer} must be sure that the crashed agent it met is not its own crashed {\tt token}.  If the {\tt token} survived until the return, 
the {\tt explorer} goes again to the place where it interrupted exploration and continues it. Whenever the {\tt explorer} realizes that its {\tt token} crashed, or the {\tt token} realizes that all its {\tt explorers}
crashed, it transits back to state {\tt cruiser}. (Now trying to find a crashed {\tt explorer} would be impossible, as it cannot be recognized.) If exploration succeeds, the {\tt explorer} and the {\tt token} learn an upper bound on the size of the graph and transit to state {\tt gatherer}. If a {\tt cruiser} meets a {\tt gatherer}, it adopts its upper bound on the size of the graph and transits to state {\tt gatherer}.

After learning the size of the graph, an agent transits to state {\tt gatherer}. In this state, an agent behaves similarly as in the previous algorithm: it first executes  procedure $R^*$ to meet all slow agents, and
then procedure $Meeting$ to meet the remaining good agents. The latter is done using the learned upper bound on the size of the graph to estimate the time after which the agent can be certain to have met all other good agents. However, there is an additional difficulty to be overcome. At the end of this period of executing $Meeting$ in state {\tt gatherer} (we will call it again the red period), an agent must be certain
that all good agents identify the same agent at whose starting node they will gather. Now this is more difficult because one agent can meet the agent with the smallest label before it crashed, and another agent may meet it after it crashed, and hence may be unable to learn the smallest label. Hence, the red period of all agents must be sufficiently long to ensure that all agents learned (possibly indirectly)
the labels, speeds and routes to the initial positions of the same subset of agents. Then they can correctly choose the same (smallest) agent in the subset (even if by this time it already crashed) and meet at its starting node. The actions of a good agent after its red period are similar as for gathering with motion faults.

We now give a detailed description of the algorithm, explaining the actions of the agents in each of the four states. In the following description, the executing agent is called $A$, has label $L$ and speed~$\theta$.

State {\tt cruiser}

Upon wakeup, agent $A$ is in this state. It performs procedure {\em Meeting} until it meets another agent in state {\tt cruiser} or an agent in state {\tt gatherer}.

In the first case, let $L_1,\dots ,L_k$ be the agents in state {\tt cruiser} that $A$ met in the first meeting, and suppose that it did not meet any agent in state {\tt gatherer} during this meeting (simultaneous meeting of many agents is possible).
Let $L'=\min(L,L_1,\dots ,L_k)$. If $L=L'$, agent $A$ transits to state {\tt token}, otherwise agent $A$ transits to state {\tt explorer}, and considers the agent with label $L'$ as its token.

In the second case, agent $A$ completes the currently executed edge traversal, adopts the upper bound $m$ on the size of the graph from the agent in state {\tt gatherer},  and transits to state {\tt gatherer}.

State {\tt explorer}

Agent $A$ executes the following procedure $ESST^*$ that results from procedure $ESST$ using the two above mentioned modifications.

{\bf Procedure}  $ESST^*$

The agent executes the consecutive steps of the modified procedure $ESST$, as in Algorithm {\tt Gathering with motion faults}. Call these steps {\em green}.
When the agent makes a green step after which it sees a crashed agent, it  goes back to its {\tt token} following the trajectory $\overline{T}$, where $T$ is the trajectory consisting of all green steps made
in the current phase until
this point. (If the meeting with the crashed agent was inside an edge, agent $A$ completes the current edge traversal and then follows $\overline{T}$). If the {\tt token} is crashed, agent $A$ transits to state {\tt cruiser}. Otherwise it follows trajectory $T$ getting back to the crashed agent. All these steps from the crashed agent to the {\tt token} and back are called {\em blue}. Then the agent makes the next (green) step of the modified procedure $ESST$, as in Algorithm {\tt Gathering with motion faults}.

Upon completion of procedure $ESST^*$, the agent traversed all edges of the graph, knows it, and adopts the total number of edge traversals $m$ as an upper bound on the size of the graph. Agent $A$
gives the integer $m$ to its token and transits to state {\tt gatherer}.

Note that for every green step made in phase $i$ of procedure $ESST^*$, the agent made less than $2g_i$ blue steps, where $g_i$ is the total number of green steps in phase $i$. Hence the total number of steps made by the agent in phase $i$ is at most $3g_i^2$. Since $g_i=2\pi_i$, the total number of steps in phase $i$ is at most $12\pi_i^2$.

State {\tt token}

Agent $A$ considers all agents that it met when they were in state {\tt cruiser} at the moment of the first meeting (and that transited to state {\tt explorer} upon this meeting) as its {\tt explorers}. Agent $A$  continues its current edge traversal and waits until one of the following two events happens for each of its {\tt explorers}. Either an {\tt explorer} completes procedure $ESST^*$ giving agent $A$ an upper bound $m$ on the size of the graph, or an {\tt explorer} $B$ is not back after the scheduled time at most  $12\pi_i^2/\beta$ in some phase $i$, where $\beta$ is the speed of $B$ (learned by $A$). As before,  there may be potentially two different upper bounds given to $A$ by {\tt explorers} that successfully completed
procedure $ESST^*$.
When one of the above events occurs for every {\tt explorer} of the agent $A$, there are two cases. The first case is when at least one {\tt explorer} managed to complete procedure $ESST^*$ and gave an upper bound to $A$. In this case, agent $A$ adopts the minimum of the (possibly two) obtained upper bounds as its upper bound on the size of the graph and transits to state {\tt gatherer}. The second case is when all {\tt explorers} of $A$ failed to  successfully complete the procedure (because they all crashed). Agent $A$ learns this and transits to state {\tt cruiser}.

State {\tt gatherer}

In this state, agent $A$ knows an upper bound $m$ on the size of the graph. It starts by executing the procedure $R^*(m,v)$, where $v$ is the node in which $A$ transited to this state
(see the description of state {\tt gatherer} in Section 3).
Recall that $\theta$ is the speed of the agent, and let $c$ be the number of edge traversals performed by agent $A$ in the execution of this procedure. Let $\gamma=\frac{\theta}{c}$. As before, upon completion of $R^*(m,v)$, agent $A$ has achieved the following two goals: 
it met all agents whose speed is less than $\gamma$, and it learned that all agents are awake. It also follows that at this time, agent $A$ learned the lower bound $\delta=\min(\gamma,\alpha)$ on the speeds of all good agents, where $\alpha$ is the minimum speed of all non-crashed agents it met.

Next, agent $A$ executes procedure $Meeting$ for time $\tau=4\nu \Pi(\nu, \log L)/\delta+12\nu\rho^2(\nu)/\delta+\nu^2P(\nu)/\delta$,
where $\nu =12\rho^2(m)$.
Call this period of time the {\em red} period of agent $A$.
During its red period, agent $A$ maintains a {\em bag} containing {\em data}  of each agent that agent $A$ met or about which it heard.
For each agent in the bag, data concerning it are: its label, its speed, the route to its starting node, and the lowest speed of any agent it is aware of. Moreover, for an agent in the bag in state {\tt gatherer}, the data includes also its adopted upper bound on the size of the graph and the lower bound $\delta$ calculated by this agent.
In the beginning of the red period the bag of $A$ consists of the data for agents that $A$ met previously, and that were not yet crashed by the time of the meeting. Whenever $A$ meets a non-crashed agent $B$, it adds the
content of the bag of $B$ to its own bag. At each step, agent $A$ updates the routes to the starting node of each agent in its bag, hence at each step these are routes from the current node.

We will show that at the end of its red period, the bag of each good agent contains the data of the same subset $\Sigma$ of agents (some of which may have already crashed by this time), and that this subset contains all good agents.

Let $m^*$ be the largest bound $m$ calculated by any agent from the set $\Sigma$. Hence $\nu^*=12\rho^2(m^*)$ is an upper bound on the values of $\nu$ calculated by all agents from the set $\Sigma$.
Let $\delta^*$ be the minimum over all values $\delta$ and all speeds that appear in the common bag of agents from the set $\Sigma$. 
Let $\psi=(12\nu^*+1)\Pi(\nu^*,\Lambda^*)/\delta^*+108\nu^*\rho^2(\nu^*)/\delta^*+5(\nu^*)^2P(\nu^*)/\delta^*$.
We will show that $\psi$ is an upper bound on the termination time of red periods of all agents in the set $\Sigma$. Let $\kappa^*$ be the largest speed that appears in the common bag of agents from the set $\Sigma$. Note that at the end of its red period, each good agent can compute $\nu^*$, $\Lambda^*$, $\delta^*$, $\kappa^*$ and $\psi$, looking at the content of its bag. 
Agent $A$ goes to the starting node of the agent whose label is the smallest in its bag, and starts its final waiting period of length $\tau'=\psi+2\psi\kappa^*/\delta^*$.

Notice that this final waiting period has the same length for all good agents, as it does not depend on the label and the speed of an individual agent.
Agent $A$ can calculate this waiting time because, when it starts it, it knows all parameters concerning all agents in its bag.

Similarly as before, we will prove that after the final waiting period, all good agents must be at the same node, which is the starting node of the smallest agent in the bag of each of them, and that agent $A$ learns this fact. After waiting time $\tau'$ (or at the time when some other agent declares gathering, if this happens earlier), agent $A$ declares that all good agents are gathered, and terminates.

The proof of correctness and the performance analysis of Algorithm  {\tt Gathering with total faults} are split into a series of lemmas.
Let $\epsilon^*$ be the minimum speed of any agent that any good agent ever becomes aware of. Let $\Lambda^*$ be the length of the longest label of any agent that any good agent ever becomes aware of. (There may exist very slow agents with speed smaller than $\epsilon^*$, or agents with very long labels, longer than $\Lambda^*$ that crash before meeting any other agent, and thus, in the model of total faults, no good agent would be aware of such a low speed or long label.)

\begin{lemma}\label{first gatherer}
By the time $2n(2\Pi(n,\Lambda^*)+24\rho^2(n))/\epsilon^*$ some good agent must transit to state {\tt gatherer}.
\end{lemma}

\begin{proof}
Consider the earliest woken good agent $S_0$ and some other good agent $S$. By the time $\Pi(n,\Lambda^*)/\epsilon^*$, agent $S$ must be woken up, for otherwise, agent $S_0$ would meet it and wake it up because it would perform at least $\Pi(n,\Lambda^*)$ edge traversals executing procedure $Meeting$. Until one of these agents
transits to state {\tt gatherer}, each of them must be in one of the states {\tt cruiser}, {\tt explorer}, or {\tt token}. In state {\tt cruiser}, an agent executes procedure $Meeting$,
and in state {\tt explorer} or {\tt token} it is involved in procedure $ESST^*$, either as an explorer or as a token. Each execution of $ESST^*$ in which a good agent is involved lasts at most the time
$12\rho^2(n)/\epsilon^*$. Hence before transiting to state {\tt gatherer}, a good agent executes procedure $Meeting$ in state {\tt cruiser} with interruptions of duration at most $12\rho^2(n)/\epsilon^*$.
There can be at most $n$ such interruptions by the time a good agent transits to state {\tt gatherer} because the only reason for  transiting from state {\tt explorer}
or {\tt token} to state {\tt cruiser} is a crash of some new agent, and there are at most $n$ agents altogether. By the time $n(2\Pi(n,\Lambda^*)+24\rho^2(n))/\epsilon^*$
after the wakeup of agent $S$ (if no good agent transited to state {\tt gatherer} before)
there must be a period of length at least $\Pi(n,\Lambda^*)/\epsilon^*$ when both agents $S$ and $S_0$ execute procedure $Meeting$. Hence they must meet because
each of them performed at least $\Pi(n,\Lambda^*)$ edge traversals executing this procedure. Then one of them transits to state {\tt explorer}, and the other one to state
{\tt token}. Since they are both good, after an additional time at most $12\rho^2(n)/\epsilon^*$, they both transit to state {\tt gatherer}. This happens by the time
$2n(2\Pi(n,\Lambda^*)+24\rho^2(n))/\epsilon^*$.
\end{proof}

Let $m$ be the upper bound on the size of the graph, adopted by the first good agent that transited to state {\tt gatherer}.

\begin{lemma}\label{all gatherers}
By the time  $4n(2\Pi(n,\Lambda^*)+24\rho^2(n))/\epsilon^* + 2m^2P(m)/\epsilon^*+ \Pi(n,\Lambda^*)/\epsilon^*$  all good agents transit to state {\tt gatherer}.
\end{lemma}

\begin{proof}
Let $S$ be the first good agent that transited to state {\tt gatherer}. In view of Lemma \ref{first gatherer}, this happened by the time $2n(2\Pi(n,\Lambda^*)+24\rho^2(n))/\epsilon^*$.
After an additional time at most $2m^2P(m)/\epsilon^*$, agent $S$ completed procedure $R^*$, and started its red period. Also, after this time, all agents are woken up.
Consider any other good agent $S'$. After time $2n(2\Pi(n,\Lambda^*)+24\rho^2(n))/\epsilon^* + 2m^2P(m)/\epsilon^*+ 2n(2\Pi(n,\Lambda^*)+24\rho^2(n))/\epsilon^*$,
agent $S$ executes procedure $Meeting$ in state {\tt gatherer} and agent $S'$ executes procedure $Meeting$ in state {\tt cruiser} (if it has not transited to state {\tt gatherer} before). Hence, after an additional time at most $\Pi(n,\Lambda^*)/\epsilon^*$ they must meet, and agent $S'$ transits to state {\tt gatherer}.
\end{proof}

\begin{lemma}\label{total red}
By the end of its red period a good agent $S$ knows that the bags of all good agents contain the data of the same subset $\Sigma$ of agents, and that subset $\Sigma$ contains
all good agents. This happens by the time $4\nu(2\Pi(\nu,\Lambda^*)+24\rho^2(\nu))/\epsilon^* + 4\nu^2P(\nu)/\epsilon^*+ \Pi(\nu,\Lambda^*)/\epsilon^*+
4\nu \Pi(\nu, \Lambda^*)/\delta+12\nu\rho^2(\nu)/\delta+\nu^2P(\nu)/\delta$.
\end{lemma}

\begin{proof}
Consider a good agent $S$ with label $L$ which started its red period at time $t$. Suppose that $m$ is the upper bound on the size of the graph, adopted by agent $S$. 
Agent $S$ knows the lower bound $\delta$ on the speeds of all good agents, and knows some upper bound on the largest upper bound $\mu$ on the size of the graph adopted by any agent.
 Indeed,  any agent $T$ adopts as its upper bound the number of edge traversals in its execution of procedure $ESST^*$. This number is at most $12\rho^2(n)$. Hence it is at most $\nu=12\rho^2(m)$, which $S$ knows.  Thus $\mu \leq \nu$.
 
 Agent $S$ knows that by the time $t$ all agents are woken up. Hence it knows that by the time
 $t+\nu(\Pi(\nu, \log L)/\delta+ \Pi(\nu, \log L)/\delta+ 12\rho^2(\nu)/\delta)$, all good agents transited to state {\tt gatherer}. Indeed, consider a good agent $S'$ that has not yet transited to state {\tt gatherer}. Hence, after time $t$, it executes procedure $Meeting$ in state {\tt cruiser}, with possible interruptions by being involved in procedure $ESST^*$, either in state {\tt explorer} or in state {\tt token}. Each period when $S'$ executes $Meeting$ in state {\tt cruiser} (between interruptions) lasts at most time
$ \Pi(\nu, \log L)/\delta$ because by this time $S'$ would meet agent $S$. Each interruption when $S'$ executes procedure $ESST^*$ in state {\tt explorer} can last time at most $12\rho^2(\nu)/\delta$ because $S'$ has speed at least $\delta$. Each interruption when $S'$ executes procedure $ESST^*$ in state {\tt token} can last at most time
$ \Pi(\nu, \log L)/\delta$ because during such interruption agent $S'$ is inside a single edge and hence it would meet agent $S$. Since there are less than $\nu$
interruptions, as each of them is due to a crash of some agent, and there are at most $\nu$ agents, we conclude that $S$ knows that by the time
 $t+\nu(\Pi(\nu, \log L)/\delta+ \Pi(\nu, \log L)/\delta+ 12\rho^2(\nu)/\delta)$, all good agents transited to state {\tt gatherer}. 
 
 Hence $S$ knows that by the time  $t'=t+\nu(\Pi(\nu, \log L)/\delta+ \Pi(\nu, \log L)/\delta+ 12\rho^2(\nu)/\delta)+ \nu^2P(\nu)/\delta$ every good agent completed the procedure $R^*$ and started its red period. Consider the time $t''=t'+\nu(2\Pi(\nu, \log L)/\delta)$. The time segment of length $\nu(2\Pi(\nu, \log L)/\delta)$ that starts at time $t'$ can be divided into $\nu$ time segments of length $2\Pi(\nu, \log L)/\delta$. In at least one of them no agent crashes. This segment is divided into two segments
 $I_1$ and $I_2$, each of length $\Pi(\nu, \log L)/\delta$. Let $\Sigma$ be the set of agents that have not crashed by the beginning of segment $I_1$.
 During the time segment $I_1$ agent $S$ meets all agents from the set $\Sigma$, and hence its bag at the end of $I_1$ is the union of the bags of all the agents from $\Sigma$.
 During segment $I_2$ agent $S$ meets again all  agents from the set $\Sigma$, and hence the bag of each of them at the end of segment $I_2$ is equal to the bag
 of agent $S$. After the end of segment $I_2$ the bags of agents do not change. The set $\Sigma$ contains all good agents. We have $t''=t+\tau$, where $\tau$ is the
 length of the red period of agent $S$. Hence by the end of the red period of agent $S$, i.e., by the time $t''$, agent $S$ knows that the bags of all good agents contain the data of the same subset $\Sigma$ of agents, and that subset $\Sigma$ contains
all good agents. Since $t\leq 4\nu(2\Pi(\nu,\Lambda^*)+24\rho^2(\nu))/\epsilon^* + 4\nu^2P(\nu)/\epsilon^*+ \Pi(\nu,\Lambda^*)/\epsilon^*$ and $\tau\leq 4\nu \Pi(\nu, \Lambda^*)/\delta+12\nu\rho^2(\nu)/\delta+\nu^2P(\nu)/\delta$, the time estimate in the  lemma follows.
\end{proof}

Let $m^*$ be the largest bound $m$ calculated by any agent from the set $\Sigma$. Hence $\nu^*=12\rho^2(m^*)$ is an upper bound on the values of $\nu$ calculated by all agents from the set $\Sigma$.
Let $\delta^*$ be the minimum over all values $\delta$ and all speeds that appear in the common bag of agents from the set $\Sigma$. Hence $\delta^* \leq \epsilon^*$. 
Let $\psi=(12\nu^*+1)\Pi(\nu^*,\Lambda^*)/\delta^*+108\nu^*\rho^2(\nu^*)/\delta^*+5(\nu^*)^2P(\nu^*)/\delta^*$.
It follows from Lemma \ref{total red} that $\psi$ is an upper bound on the termination of red periods of all agents from the set $\Sigma$. 
(Note that if an agent from the set $\Sigma$ crashed before completing its red period, this red period could not extend beyond the time established in Lemma \ref{total red} for good agents.)
Let $\kappa^*$ be the largest speed that appears in the common bag of agents from the set $\Sigma$.

\begin{lemma}\label{destin total all}
 By the time $\psi+2\psi\kappa^*/\delta^*$ any good agent $S$ is at the starting node of the smallest agent $S^*$ from the set $\Sigma$, and starts its final waiting period. 
 \end{lemma}
 
 \begin{proof}
 Let $T$ be the trajectory of agent $S$ until the end of its red period, and let $T'$ be the trajectory from the initial node of agent $S^*$ to some node $x$ of the trajectory $T$, learned by agent $S$ from its bag.
 (Notice that agent $S^*$ might have crashed and this trajectory may sometimes be learned by agent $S$ indirectly.) The length of each of the trajectories $T$ and $T'$ is at most $\psi\kappa^*$. After the end of its red period, agent $S$ knows both these trajectories. Hence it can backtrack to node $x$ and then backtrack to the starting node of agent $S^*$. This takes at most $2\psi\kappa^*$ edge traversals, and hence it takes time at most $2\psi\kappa^*/\delta^*$. Hence by the time 
 $\psi+2\psi\kappa^*/\delta^*$
 agent $S$ is at the starting node of agent  $S^*$ and starts its final waiting period.
  \end{proof}

 \begin{lemma}\label{total all}
By the time $2\psi+4\psi\kappa^*/\delta^*$  all good agents are at the starting node of the smallest agent from the set $\Sigma$, 
all of them know it, and correctly declare that gathering is completed.
\end{lemma}

\begin{proof}
We first show that at the beginning of its final waiting period, every good agent can compute the length $\tau'=\psi+2\psi\kappa^*/\delta^*$ of this period.
Indeed this value depends on four parameters: $\Lambda^*$, $\kappa^*$, $\nu^*$ and $\delta^*$.
At the beginning of its final waiting period, every good agent can compute these parameters from its bag (that is the same for all good agents). 
Hence it can compute $\psi$ and finally $\tau'$.

In view of Lemma \ref{destin total all}, by the time $\psi+2\psi\kappa^*/\delta^*$, every good agent is at the starting node of the smallest agent $S^*$ from the set $\Sigma$, and starts its final waiting period.
Hence at the end of its final waiting period, which has the same length $\tau'=\psi+2\psi\kappa^*/\delta^*$, every good agent knows that all good agents started their final waiting period.
Hence, by the time $2\psi+4\psi\kappa^*/\delta^*$, every good agent correctly declares that gathering is completed.
\end{proof}

We are now ready to prove the main result of this section.

\begin{theorem}
Algorithm {\tt Gathering with total faults} is a correct algorithm for the task of gathering with total faults. It works in time polynomial in the size $n$ of the graph, in the length $\Lambda$ of the largest label,
 in the inverse of the smallest speed $\epsilon$, and in the ratio $r$ between the largest and the smallest speed.
\end{theorem}

\begin{proof}
The correctness of the algorithm follows immediately from Lemma \ref{total all}. 
It remains to verify that the time $2\psi+4\psi\kappa^*/\delta^*$  is polynomial in $n$, $\Lambda$, $1/\epsilon$ and $r$. 

We have $\kappa^* \leq \kappa$. The value $1/\delta$ computed by each agent is  polynomial in
$n$ and in $1/\epsilon$, and thus so is $1/\delta ^*$. Hence $\kappa^*/\delta^*$ is polynomial in $r=\kappa/\epsilon$, in $n$ and in $1/\epsilon$.
Thus it is enough to verify that $\psi=(12\nu^*+1)\Pi(\nu^*,\Lambda^*)/\delta^*+108\nu^*\rho^2(\nu^*)/\delta^*+5(\nu^*)^2P(\nu^*)/\delta^*$ is  polynomial in  $n$, $\Lambda$ and $1/\epsilon$.

The upper bound $m$ on the size of the graph, calculated by each agent, is at most $12\rho^2(n)$.
Since $\rho$ is a polynomial, it follows that $m$ is polynomial in $n$, and so is $m^*$. Hence $\nu^*=12\rho^2(m^*)$ is also polynomial in $n$.

We have $\Lambda^* \leq \Lambda$, and hence each of the three summands in $\psi$ are polynomial in $n$, $\Lambda$ and $1/\epsilon$ because $\Pi$, $\rho$ and $P$ are polynomials. Hence
$\psi$ is polynomial in $n$, $\Lambda$ and $1/\epsilon$.
\end{proof}

\section{Conclusion}

Our algorithm for motion faults works for any team of at least two agents, regardless of the number of faulty agents, and our algorithm for total faults works for any team of agents, containing at least two good agents; the latter is necessary. Hence, in terms of feasibility of gathering with crash faults, our results are optimal.

We showed that the time of our algorithms is polynomial in  the size $n$ of the graph, in the length $\Lambda$ of the largest label,
 in the inverse of the smallest speed $\epsilon$, and in the ratio $r$ between the largest and the smallest speed. Optimizing this polynomial seems to be out of the reach at present, as even the optimal rendezvous time of two good agents is not known for arbitrary graphs. However, it remains open if the time of gathering with crash faults really has to depend on all four above parameters. The dependence on $n$ and on $1/\epsilon$ seems natural. The time of gathering should somehow depend on the labels as well, but maybe it could depend on the length of the smallest label (as it is the case for rendezvous without faults) and not the largest. The most problematic seems
 the dependence on the ratio between the largest and the smallest speed. Is this dependence necessary?

\bibliographystyle{plain}


\end{document}